\RequirePackage[l2tabu, orthodox]{nag}
\RequirePackage{snapshot}

\documentclass[10pt,onecolumn]{extarticle}

\makeatletter
\if@twocolumn
  \usepackage[dvips,a4paper,top=0.5in, bottom=0.5in, left=0.75in, right=0.5in,includefoot,heightrounded]{geometry}
\else
  \usepackage[dvips,a4paper,margin=1in,includefoot,heightrounded]{geometry}
\fi

\usepackage{srcltx}

\usepackage[russian,portuges,english]{babel}

\iflanguage{portuges}
    {\newcommand{\keywordname}{Palavras-chaves}}
    {\newcommand{\keywordname}{Keywords}}

\usepackage{amsmath}
\usepackage{amssymb}
\usepackage{amsfonts}
\usepackage{mathtools}
\usepackage{amsthm}

\usepackage{abstract}

\usepackage{graphicx}
\usepackage[usenames,dvipsnames,svgnames,x11names]{xcolor}
\usepackage{subfigure}

\usepackage{booktabs}

\usepackage{setspace}
\usepackage{flushend}

\usepackage{multicol}

\usepackage{cite}

\usepackage{hyperref}
\usepackage[normalem]{ulem}
\usepackage{soul}

\usepackage{enumerate}

\usepackage[noend]{algpseudocode}

\usepackage{algorithm}
\floatname{algorithm}{Listing}

\usepackage{listings}

\usepackage[short,12hr]{datetime}

\usepackage{siunitx}
\usepackage{draftcopy}

       \usepackage{fouriernc}
\newtheorem{proposition}{Proposition}

\newtheorem{lemma}{Lemma}
\newtheorem{corollary}{Corollary}

\newcommand*{\transp}{^{\mathsf{T}}}
\newcommand*{\hermit}{^{\mathsf{H}}}

\usepackage{mleftright}

\newcommand{\iverson}[1]{\ensuremath{ \mleft\llbracket#1\mright\rrbracket }}

\newcommand{\printtitle}{%
\makeatletter
\if@twocolumn

\twocolumn[%
  \maketitle
  \begin{onecolabstract}
    \myabstract
  \end{onecolabstract}
  \begin{center}
    \small
    \textbf{\keywordname}
    \\\medskip
    \mykeywords
  \end{center}
  \bigskip
]
\saythanks
\else
  \maketitle
  \begin{onecolabstract}
    \myabstract
  \end{onecolabstract}
  \begin{center}
    \small
    \textbf{\keywordname}
    \\\medskip
    \mykeywords
  \end{center}
  \bigskip
  \onehalfspacing
\fi
\makeatother
}

\author{%
R.~J.~Cintra%
\thanks{%
R. J. Cintra is with the
Industrial Signal Processing Laboratory,
Department of Technology,
UFPE, Brazil.
E-mail: \url{rjdsc@de.ufpe.br}}
}

\title{%
A Note on the Orthogonalization of Real-valued Trigonometrical Basis Functions}

\newcommand{\myabstract}{%
Adopting an intuitive approach,
this letter
shows
the derivation
of three orthogonal real matrices
based on the Fourier matrix.
Only elementary methods are employed.
}

\newcommand{\mykeywords}{%
Discrete Fourier transform,
orthogonalization,
basis functions,
discrete cosine transform.
}

\date{}

\begin{document}

\printtitle

\section{Introduction}

The discrete Fourier transform (DFT)
is
a central tool in
signal processing,
numerical analysis,
and
statistics,
among other fields.
Because the DFT is defined over the complex field~$\mathbb{C}$,
its
implementation
requires complex arithmetic,
which might lead to increased computational complexity
in fast algorithms~\cite{blahut2010fast}.
Literature archives several analysis methods
aiming
at circumventing
the need for complex calculations.
Probably,
the two most prominent real-valued methods
are
(i)~the discrete Hartley transform (DHT)
and
(ii)~the discrete cosine transform (DCT),
in particular the DCT of type~II~\cite{ahmed1974discrete}.

The DHT formulation is quite simple and elegant.
Indeed,
the DHT can be expressed
in terms of the real and imaginary parts of the DFT.
Conversely, the DFT
is obtained from
the even and odd parts
of the DHT~\cite[p.~29]{bracewell1986hartley}.
The $N$-point DHT kernel
is simply
$\operatorname{cas}\left( \frac{2\pi}{N}k n\right)$,
$k,n=0,1,\ldots,N-1$,
where
$\operatorname{cas}(x) = \cos(x)+\sin(x)$~\cite[p.~27]{bracewell1986hartley}.

On the other hand,
DCT
has been defined and derived by several
approaches in literature.
The various methods
for obtaining the DCT
encompass:
(i)~the discretization of the Fourier cosine transform~\cite{kitajima1980symmetric};
(ii)~the properties of Chebyshev polynomials~\cite{ahmed1974discrete,ahmed1975orthogonal,kitajima1980symmetric};
(iii)~the discrete solution of undamped harmonic oscillator equation~\cite{strang1999discrete};
(iv)~the construction of periodic, even symmetric sequences~\cite[p.~589]{oppenheim1999discrete};
(v)~the limiting case of the Karhunen-Lo\`eve transform~\cite{flickner1982derivation},
and
(vi)~the diagonalization of the autocovariance matrix of highly correlated,
stationary, Markov type-1 signals~\cite[p.~28]{rao1990discrete}.
In contrast to the DHT case,
the DCT derivation tends to be more sophisticate.

Inspired by~\cite{kay2006intuitive}
and
intended to the undergraduate readers,
in this letter,
we
aim at the following goal.
Based on the discrete Fourier transform basis functions,
we seek to derive a real, orthogonal transformation matrix
by
(i)~using simple mathematical methods
and
(ii)~not requiring any matrix property besides orthogonality itself.
By simple methods,
we subjectively mean
the set of commonly available tools
to first-year undergraduate students in Engineering,
such as linear algebra as in~\cite{strang2006linear}.

\section{Background and Notation}

\subsection{Notation and Definitions}

We adopt the following notation.
Transposition and conjugate transposition
are denoted by $\transp$ and $\hermit$, respectively.
The Iverson bracket
$\iverson{P}$ returns 1, if $P$ is true; 0, otherwise~\cite{knuth1992two}.
Vectors are column vectors.
Vectors and matrices are typed in boldface.
The $N$-point inner product
of vectors
$\mathbf{u}$ and $\mathbf{v}$
is $\langle \mathbf{u}, \mathbf{v} \rangle_N
=
\sum_{n=0}^{N-1} u[n]\cdot (v[n])\hermit$.
A basis
$\{ \mathbf{u}_0, \mathbf{u}_1, \ldots, \mathbf{u}_{N-1} \}$
is orthogonal if
$\langle \mathbf{u}_i, \mathbf{u}_k \rangle = 0$,
for $i\neq k$.
An orthogonal basis is orthonormal
if
$
\langle \mathbf{u}_i, \mathbf{u}_k \rangle = \iverson{i=k}
$~\cite[p.~32]{hsu1969vector}.
The imaginary unit is $j \triangleq \sqrt{-1}$.
Congruence is denoted by $b\equiv a \pmod m$
which means $b-a = k\cdot m$,
where $k$ is an integer.

\subsection{Useful Identities}

\begin{lemma}
\label{lemma-sums}
\begin{enumerate}[(i)]
Let $k$ and $N\neq0$ be integers.
\item
Then:
\begin{align}
\sum_{n=0}^{N-1}
\cos\left(
\frac{2\pi}{N}
kn
\right)
=
N \cdot \iverson{k\equiv 0 \pmod{N}}
.
\end{align}

\item
If $N$ is even, then it holds true that:
\begin{align}
\sum_{n=0}^{\frac{N}{2}-1}
\cos\left(
\frac{2\pi}{N}
kn
\right)
=
\frac{N}{2}
\cdot
\iverson{k\equiv 0 \pmod{N}}
+
\iverson{\text{$k$ odd}}
.
\end{align}
\end{enumerate}
\end{lemma}
\begin{proof}
(i)~See~\cite[4.4.3.13, p.~641]{prudnikov1998integrals}.

(ii)
By direct application
of~\cite[1.342.2]{gradshteyn2014table}
(Lagrange's trigonometric identity),
it follows that:
\begin{align}
\sum_{n=0}^{\frac{N}{2}-1}
\cos\left(
\frac{2\pi}{N}
kn
\right)
&
=
\frac{N}{2}
\cdot
\iverson{k\equiv 0 \pmod{N}}
+
\frac{1-(-1)^k}{2}
\cdot
\iverson{k\not\equiv 0 \pmod{N}}
\\
&
=
\frac{N}{2}
\cdot
\iverson{k\equiv 0 \pmod{N}}
+
\frac{1-(-1)^k}{2}
\cdot
(
1
-
\iverson{k\equiv 0 \pmod{N}}
)
\\
&
=
\left(
\frac{N}{2}
-
\frac{1-(-1)^k}{2}
\right)
\cdot
\iverson{k\equiv 0 \pmod{N}}
+
\frac{1-(-1)^k}{2}
.
\end{align}
Because $N$ is even,
if $k\equiv 0 \pmod{N}$,
then $k$ is also even and $\frac{1-(-1)^k}{2}=0$.
Notice that
$\frac{1-(-1)^k}{2} = \iverson{\text{$k$ odd}}$.

\end{proof}

\begin{lemma}
\label{lemma-sum-prod}

Let $i$, $k$, and $N\neq0$ be integers.
\begin{enumerate}[(i)]

\item
Then:
\begin{align}
\sum_{n=0}^{N-1}
\cos\left(
\frac{2\pi}{N}in
\right)
\cdot
\cos\left(
\frac{2\pi}{N}kn
\right)
=
\frac{N}{2}
\cdot
\left(
\iverson{i-k\equiv 0 \pmod{N}}
+
\iverson{i+k\equiv 0 \pmod{N}}
\right)
.
\end{align}

\item
If $N$ is even, then it holds true that:
\begin{align}
\label{equation-lemma-2-ii}
\sum_{n=0}^{\frac{N}{2}-1}
\cos\left(
\frac{2\pi}{N}in
\right)
\cdot
\cos\left(
\frac{2\pi}{N}kn
\right)
=
&
\frac{N}{4}
\cdot
\iverson{i-k\equiv 0 \pmod{N}}
+
\frac{N}{4}
\cdot
\iverson{i+k\equiv 0 \pmod{N}}
+
\iverson{\text{$(i+k)$ odd}}
.
\end{align}

\end{enumerate}
\end{lemma}
\begin{proof}
Both results follow from the application of Lemma~\ref{lemma-sums}
to the identity:
$
\cos\left(
\frac{2\pi}{N}i\cdot n
\right)
\cdot
\cos\left(
\frac{2\pi}{N}k\cdot n
\right)
=
\frac{1}{2}
\cdot
  \cos\left(
  \frac{2\pi}{N} (i-k) n
  \right)
  +
\frac{1}{2}
\cdot
  \cos\left(
  \frac{2\pi}{N} (i+k) n
  \right)
$.
(i)~See~\cite[4.4.4.6--8, p.~641]{prudnikov1998integrals}.
(ii)~Notice that
$\iverson{\text{$(i-k)$ odd}} = \iverson{\text{$(i+k)$ odd}}$.
The zeros for the expression in item~(ii)
are in~\cite[4.4.4.5, p.~641]{prudnikov1998integrals}.
\end{proof}

\begin{corollary}
If $0\leq i,k \leq \frac{N}{2}-1$,
then~\eqref{equation-lemma-2-ii}
becomes:
\begin{align}
\sum_{n=0}^{\frac{N}{2}-1}
\cos\left(
\frac{2\pi}{N}in
\right)
\cdot
\cos\left(
\frac{2\pi}{N}kn
\right)
=
&
\frac{N}{4}
\cdot
\iverson{i-k\equiv 0 \pmod{N}}
+
\frac{N}{4}
\cdot
\iverson{i=k\equiv0 \pmod{N}}
+
\iverson{\text{$(i+k)$ odd}}
.
\end{align}

\end{corollary}
\begin{proof}
If $i,k = 0,1,\ldots,\frac{N}{2}-1$,
then
$\iverson{i+k\equiv 0 \pmod{N}}=1$
only when $i=k\equiv0 \pmod{N}$.

\end{proof}

\subsection{The Discrete Fourier Transform}

The discrete Fourier transform relates
an $N$-point input discrete-time signal
$x[n] \in \mathbb{C}$, $n=0,1,\ldots, N-1$,
to
an $N$-point frequency-domain discrete signal
$X[k] \in \mathbb{C}$, $k=0,1,\ldots, N-1$,
according to~\cite[p.~561]{oppenheim1999discrete}:
\begin{align}
X[k]
=
\sum_{n=0}^{N-1}
x[n]
\cdot
\exp\left(
-j \frac{2\pi}{N} k n
\right)
,
\quad
k = 0, 1, \ldots, N-1
.
\end{align}
The DFT coefficient $X[k]$
can be seen
as the
inner product
of
the input signal (or vector)
$\mathbf{x} = \begin{bmatrix}x[0] & x[1] & \cdots & x[N-1]\end{bmatrix}\transp$
and
the
basis function (or vector)
$
\mathbf{f}_k =
\Big[
\exp\left(
j \frac{2\pi}{N} k n
\right)
\Big]_{n = 0,1,\ldots,N-1}\transp
$.
Thus,
we can write:
$
X[k]
=
\langle
\mathbf{x}
,
\mathbf{f}_k
\rangle_N
$.
The coefficients
$X[k]$,
$k=0,1,\ldots,N-1$,
can be grouped in a vector,
$\mathbf{X} = \begin{bmatrix}X[0] & X[1] & \cdots & X[N-1]\end{bmatrix}\transp$,
and
expressed in matrix format as
$
\mathbf{X}
=
\mathbf{F}_{N}
\cdot
\mathbf{x}
$,
where
$
\mathbf{F}_N
=
\begin{bmatrix}
\mathbf{f}_0 &
\mathbf{f}_1 &
\cdots &
\mathbf{f}_{N-1}
\end{bmatrix}\hermit
$
is the DFT matrix.
The vectors
$\mathbf{f}_k$,
$k = 0, 1, \ldots, N-1$,
can be interpreted as
discrete-time signals
defined over the interval
$n=0,1,\ldots, N-1$:
\begin{align}
f_k[n]
=
\exp\left(
j \frac{2\pi}{N} k n
\right)
,
\quad
k = 0, 1, \ldots, N-1
.
\end{align}

Vectors $\mathbf{f}_k$
(rows of $\mathbf{F}_N$)
constitute
an
orthogonal
basis for $\mathbb{C}^N$.
Indeed,
we have that:
\begin{align}
\langle
\mathbf{f}_i
,
\mathbf{f}_k
\rangle_N
=
&
\sum_{n=0}^{N-1}
f_{i}[n]
\cdot
f_{k}\hermit[n]
\\
&
=
\sum_{n=0}^{N-1}
\exp\left(
j\frac{2\pi}{N} i n
\right)
\cdot
\exp\left(
-j\frac{2\pi}{N} k n
\right)
\\
&
=
\sum_{n=0}^{N-1}
\exp\left(
j\frac{2\pi}{N} (i-k) n
\right)
\\
\label{equation-dft-ortho}
&
=
N
\cdot
\iverson{i\equiv k \pmod{N}}
,
\end{align}
which furnishes
the following
inner product table:
\begin{align}
\begin{array}{c|cccc cccc}
\langle \mathbf{f}_i, \mathbf{f}_k\rangle_{N}
&
\mathbf{f}_0 & \mathbf{f}_1 & \mathbf{f}_2 & \mathbf{f}_3 & \mathbf{f}_4 & \mathbf{f}_5 & \mathbf{f}_6 & \mathbf{f}_7
\\
\hline
\mathbf{f}_0
&
N & 0 & 0 & 0 & 0 & 0 & 0 & 0
\\
\mathbf{f}_1
&
0 & N & 0 & 0 & 0 & 0 & 0 & 0
\\
\mathbf{f}_2
&
0 & 0 & N & 0 & 0 & 0 & 0 & 0
\\
\mathbf{f}_3
&
0 & 0 & 0 & N & 0 & 0 & 0 & 0
\\
\mathbf{f}_4
&
0 & 0 & 0 & 0 & N & 0 & 0 & 0
\\
\mathbf{f}_5
&
0 & 0 & 0 & 0 & 0 & N & 0 & 0
\\
\mathbf{f}_6
&
0 & 0 & 0 & 0 & 0 & 0 & N & 0
\\
\mathbf{f}_7
&
0 & 0 & 0 & 0 & 0 & 0 & 0 & N
\\
\hline
\end{array}
\end{align}

Notice that the
periodic nature of the complex exponential
effects modular arithmetic.
While the indexes
$i$ and $k$ range from 0 to $N-1$,
the congruence
$i \equiv k \pmod N$
can be understood
$i=k$.
However,
we keep congruence instead of equality notation,
so we do not need to pay attention
to index ranges and numbers ``wrapping around''
on a case-by-case basis.

\subsection{Real Part of the DFT}

Consider the real part of the discrete Fourier transform
basis as a candidate basis for an orthogonal system.
Thus,
the candidate basis functions are:
\begin{align}
r_k[n]
=
\Re(f_k[n])
=
\cos\left(
\frac{2\pi}{N}kn
\right)
,
\quad
k,n = 0, 1, \ldots, N-1
.
\end{align}
To check whether this basis
is orthogonal,
it suffices
to compute,
for $i,k = 0, 1, \ldots, N-1$:
\begin{align}
\label{eq-real-fourier-inner-product}
\langle \mathbf{r}_i, \mathbf{r}_k\rangle_N
=
&
\sum_{n=0}^{N-1}
r_{i}[n]
\cdot
r_{k}[n]
\\
=
&
\sum_{n=0}^{N-1}
\cos\left(
\frac{2\pi}{N}in
\right)
\cdot
\cos\left(
\frac{2\pi}{N}kn
\right)
\\
=&
\frac{N}{2}
\cdot
\left(
\iverson{i-k\equiv 0 \pmod{N}}
+
\iverson{i+k\equiv 0 \pmod{N}}
\right)
.
\end{align}
It
means
that
the set of signals
$r_k[n]$,
$k = 0, 1,\ldots, N-1$,
does not form an orthogonal set,
because
$
\langle \mathbf{r}_k, \mathbf{r}_k \rangle
=
\langle \mathbf{r}_k, \mathbf{r}_{N-k} \rangle
\neq
0
$.

\paragraph{The X Pattern.}
For example,
let $N=8$.
We have that
the
inner product
calculation
furnishes
the following
``saltire''-like pattern:
\begin{align}
\label{tab-saltire}
\begin{array}{c|cccc cccc}
\langle \mathbf{r}_i, \mathbf{r}_k\rangle_N
&
\mathbf{r}_0 & \mathbf{r}_1 & \mathbf{r}_2 & \mathbf{r}_3 & \mathbf{r}_4 & \mathbf{r}_5 & \mathbf{r}_6 & \mathbf{r}_7
\\
\hline
\mathbf{r}_0
&
N & 0 & 0 & 0 & 0 & 0 & 0 & 0
\\
\mathbf{r}_1
&
0 & \frac{N}{2} & 0 & 0 & 0 & 0 & 0 & \frac{N}{2}
\\
\mathbf{r}_2
&
0 & 0 & \frac{N}{2} & 0 & 0 & 0 & \frac{N}{2} & 0
\\
\mathbf{r}_3
&
0 & 0 & 0 & \frac{N}{2} & 0 & \frac{N}{2} & 0 & 0
\\
\mathbf{r}_4
&
0 & 0 & 0 & 0 & N & 0 & 0 & 0
\\
\mathbf{r}_5
&
0 & 0 & 0 & \frac{N}{2} & 0 & \frac{N}{2} & 0 & 0
\\
\mathbf{r}_6
&
0 & 0 & \frac{N}{2} & 0 & 0 & 0 & \frac{N}{2} & 0
\\
\mathbf{r}_7
&
0 & \frac{N}{2} & 0 & 0 & 0 & 0 & 0 & \frac{N}{2}
\\
\hline
\end{array}
\end{align}

\section{Intuitive Orthogonalization}

\subsection{Real and Half-sized}
\label{section-real-half}

An approach to address the lack of orthogonality
consists of removing
from the candidate basis set
the
functions
that
make
that
the inner product computation
to
generate ambiguity.
Thus,
let us consider the new candidate basis function set
given
by:
\begin{align}
r_k[n]
=
\cos\left(
\frac{2\pi}{N}kn
\right)
,
\quad
k,n = 0, 1, \ldots, \frac{N}{2}-1
.
\end{align}
Now $k$ ranges from 0 to $\frac{N}{2}-1$.
This new candidate basis
is simply the first half the previous
candidate basis,
in the hopes
of keeping only the north-west corner
of the ``X'' pattern in~\eqref{tab-saltire}.
Notice that the functions
of this candidate set
are $\frac{N}{2}$-point long,
thus half the size of the previous candidate functions.
To verify whether this basis set
is orthogonal,
we have the following computation,
for
$i,k = 0, 1, \ldots, \frac{N}{2}-1$:
\begin{align}
\label{eq-real-half-fourier-inner-product}
\langle \mathbf{r}_i, \mathbf{r}_k\rangle_{\frac{N}{2}}
&
=
\sum_{n=0}^{\frac{N}{2}-1}
r_{i}[n]
\cdot
r_{k}[n]
\\
&
=
\sum_{n=0}^{\frac{N}{2}-1}
\cos\left(
\frac{2\pi}{N}in
\right)
\cdot
\cos\left(
\frac{2\pi}{N}kn
\right)
\\
&
=
\frac{N}{4}
\cdot
\iverson{i+k\equiv 0 \pmod{N}}
+
\frac{N}{4}
\cdot
\iverson{i \equiv k \pmod{N}}
+
\iverson{\text{$(i+k)$ odd}}
\label{equation-inner-prod-real-half}
.
\end{align}
Thus,
the inner product
$
\langle \mathbf{r}_i, \mathbf{r}_k\rangle_{\frac{N}{2}}
=
1
$,
if
$i+k$
is odd.
To better understand this effect,
consider $N=16$.
We have the following
inner product table:
\begin{align}
\label{tab-checkboard}
\begin{array}{c|cccc cccc}
\langle \mathbf{r}_i, \mathbf{r}_k\rangle_{\frac{N}{2}}
&
\mathbf{r}_0 & \mathbf{r}_1 & \mathbf{r}_2 & \mathbf{r}_3 & \mathbf{r}_4 & \mathbf{r}_5 & \mathbf{r}_6 & \mathbf{r}_7
\\
\hline
\mathbf{r}_0
&
\frac{N}{2} & 1 & 0 & 1 & 0 & 1 & 0 & 1
\\
\mathbf{r}_1
&
1 & \frac{N}{4} & 1 & 0 & 1 & 0 & 1 & 0
\\
\mathbf{r}_2
&
0 & 1 & \frac{N}{4} & 1 & 0 & 1 & 0 & 1
\\
\mathbf{r}_3
&
1 & 0 & 1 & \frac{N}{4} & 1 & 0 & 1 & 0
\\
\mathbf{r}_4
&
0 & 1 & 0 & 1 & \frac{N}{4} & 1 & 0 & 1
\\
\mathbf{r}_5
&
1 & 0 & 1 & 0 & 1 & \frac{N}{4} & 1 & 0
\\
\mathbf{r}_6
&
0 & 1 & 0 & 1 & 0 & 1 & \frac{N}{4} & 1
\\
\mathbf{r}_7
&
1 & 0 & 1 & 0 & 1 & 0 & 1 & \frac{N}{4}
\\
\hline
\end{array}
\end{align}

For instance,
$
\langle \mathbf{r}_1, \mathbf{r}_0\rangle_{\frac{N}{2}}
=
\langle \mathbf{r}_1, \mathbf{r}_2\rangle_{\frac{N}{2}}
=
\langle \mathbf{r}_1, \mathbf{r}_4\rangle_{\frac{N}{2}}
=
\langle \mathbf{r}_1, \mathbf{r}_6\rangle_{\frac{N}{2}}
\neq
0
$
.
Again,
the new candidate basis is not orthogonal,
and it seems to be even ``less orthogonal''
than the candidate basis from the previous section.

\subsection{Removing the Check Pattern I}
\label{sec-approach-1}

\subsubsection{Real, Half-sized, and Extended I}
\label{sec-real-half-ext-1}

The check pattern
is
due
to
the term
$\iverson{\text{$(i+k)$ odd}}$
in~\eqref{equation-inner-prod-real-half},
i.e.
the parity of $i+k$.
The following
quite well-known
trigonometric identity
resembles the discussed pattern:
\begin{align}
\iverson{\text{$(i+k)$ even}}
-
\iverson{\text{$(i+k)$ odd}}
&
=
\cos(\pi(i+k))
.
\end{align}
The above expression corresponds
to the following pattern:
\begin{align}
\begin{array}{c|rrrr rrrr}
i,k
&
0 & 1 & 2 & 3 & 4 & 5 & 6 & 7
\\
\hline
0
&
1 & -1 & 1 & -1 & 1 & -1 & 1 & -1
\\
1
&
-1 & 1 & -1 & 1 & -1 & 1 & -1 & 1
\\
2
&
1 & -1 & 1 & -1 & 1 & -1 & 1 & -1
\\
3
&
-1 & 1 & -1 & 1 & -1 & 1 & -1 & 1
\\
4
&
1 & -1 & 1 & -1 & 1 & -1 & 1 & -1
\\
5
&
-1 & 1 & -1 & 1 & -1 & 1 & -1 & 1
\\
6
&
1 & -1 & 1 & -1 & 1 & -1 & 1 & -1
\\
7
&
-1 & 1 & -1 & 1 & -1 & 1 & -1 & 1
\\
\hline
\end{array}
\end{align}
Now notice the following manipulation:
\begin{align}
\iverson{\text{$(i+k)$ even}}
-
\iverson{\text{$(i+k)$ odd}}
&
=
\cos(\pi(i+k))
\\
&
=
\cos(\pi i)
\cdot
\cos(\pi k)
-
\sin(\pi i)
\cdot
\sin(\pi k)
\\
&
=
\cos\left(\frac{2\pi}{N} i \frac{N}{2} \right)
\cdot
\cos\left(\frac{2\pi}{N} k \frac{N}{2} \right)
.
\end{align}
The above quantity
is
$
r_{i}\left[\frac{N}{2}\right]
\cdot
r_{k}\left[\frac{N}{2}\right]
$.
Thus,
let us extend
each basis function
to be
($\frac{N}{2}+1$)-point long
with this extra component at $n=\frac{N}{2}$.
Now let us check whether the orthogonality property is satisfied.
Considering
Lemma~\ref{lemma-sum-prod},
we have the following expressions:
\begin{align}
\langle
\mathbf{r}_i
,
\mathbf{r}_k
\rangle_{\frac{N}{2}+1}
=
&
\sum_{n=0}^{\frac{N}{2}}
r_{i}[n]
\cdot
r_{k}[n]
\\
=
&
\sum_{n=0}^{\frac{N}{2}-1}
r_{i}[n]
\cdot
r_{k}[n]
+
r_{i}\left[\frac{N}{2}\right]
\cdot
r_{k}\left[\frac{N}{2}\right]
\\
=
&
\frac{N}{4}
\cdot
\iverson{i+k\equiv 0 \pmod{N}}
+
\frac{N}{4}
\cdot
\iverson{i \equiv k \pmod{N}}
+
\iverson{\text{$(i+k)$ odd}}
\nonumber
\\
&
+
\iverson{\text{$(i+k)$ even}}
-
\iverson{\text{$(i+k)$ odd}}
\\
=
&
\frac{N}{4}
\cdot
\iverson{i+k\equiv 0 \pmod{N}}
+
\frac{N}{4}
\cdot
\iverson{i \equiv k \pmod{N}}
+
\iverson{\text{$(i+k)$ even}}
.
\label{equation-inner-prod-fail}
\end{align}
The trick did not work well.
The inner product is still
dependent
on
the parity of $i+k$,
which is given by
$\iverson{\text{$(i+k)$ even}}$.
To illustrate,
consider $N=16$.
Now we have an inverse check pattern
with ones scattered at locations of opposite parity,
as shown below:
\begin{align}
\begin{array}{c|cccc cccc c}
\langle \mathbf{r}_i, \mathbf{r}_k\rangle_{\frac{N}{2}+1}
&
\mathbf{r}_0 & \mathbf{r}_1 & \mathbf{r}_2 & \mathbf{r}_3 & \mathbf{r}_4 & \mathbf{r}_5 & \mathbf{r}_6 & \mathbf{r}_7 & \mathbf{r}_8
\\
\hline
\mathbf{r}_0
&
\frac{N}{2}+1 & 0 & 1 & 0 & 1 & 0 & 1 & 0 & 1
\\
\mathbf{r}_1
&
0 & \frac{N}{4}+1 & 0 & 1 & 0 & 1 & 0 & 1 & 0
\\
\mathbf{r}_2
&
1 & 0 & \frac{N}{4}+1 & 0 & 1 & 0 & 1 & 0 & 1
\\
\mathbf{r}_3
&
0 & 1 & 0 & \frac{N}{4}+1 & 0 & 1 & 0 & 1 & 0
\\
\mathbf{r}_4
&
1 & 0 & 1 & 0 & \frac{N}{4}+1 & 0 & 1 & 0 & 1
\\
\mathbf{r}_5
&
0 & 1 & 0 & 1 & 0 & \frac{N}{4}+1 & 0 & 1 & 0
\\
\mathbf{r}_6
&
1 & 0 & 1 & 0 & 1 & 0 & \frac{N}{4}+1 & 0 & 1
\\
\mathbf{r}_7
&
0 & 1 & 0 & 1 & 0 & 1 & 0 & \frac{N}{4}+1 & 0
\\
\mathbf{r}_8
&
1 & 0 & 1 & 0 & 1 & 0 & 1 & 0 & \frac{N}{2}+1
\\
\hline
\end{array}
\end{align}

\subsubsection{Real, Half-sized, and Extended II}
\label{sec-real-half-ext-2}

Let us try again with a small modification.
Consider the following
extension to the candidate basis functions:
\begin{align}
\check{r}_i[n]
=
\begin{cases}
r_i[n], & \text{if $n = 0, 1, \ldots, \frac{N}{2}-1$,}
\\
a
\cdot
r_i\left[\frac{N}{2}\right], & \text{if $n=\frac{N}{2}$.}
\end{cases}
\end{align}
Now the included term at $n=\frac{N}{2}$ is scaled by
a real-valued factor $a\neq 0$.
Let us repeat the check for orthogonality.
Considering~\eqref{equation-inner-prod-fail},
we have:
\begin{align}
\langle
\check{\mathbf{r}}_i
,
\check{\mathbf{r}}_k
\rangle_{\frac{N}{2}+1}
=
&
\sum_{n=0}^{\frac{N}{2}}
\check{r}_{i}[n]
\cdot
\check{r}_{k}[n]
\\
=
&
\sum_{n=0}^{\frac{N}{2}-1}
r_{i}[n]
\cdot
r_{k}[n]
+
a^2
\cdot
r_{i}\left[\frac{N}{2}\right]
\cdot
r_{k}\left[\frac{N}{2}\right]
\\
=
&
\frac{N}{4}
\cdot
\iverson{i+k\equiv 0 \pmod{N}}
+
\frac{N}{4}
\cdot
\iverson{i\equiv k \pmod{N}}
+
\iverson{\text{$(i+k)$ odd}}
\nonumber
\\
&
+
a^2
\cdot
\iverson{\text{$(i+k)$ even}}
-
a^2
\cdot
\iverson{\text{$(i+k)$ odd}}
\\
=&
\frac{N}{4}
\cdot
\iverson{i+k\equiv 0 \pmod{N}}
+
\frac{N}{4}
\cdot
\iverson{i\equiv k \pmod{N}}
\nonumber
\\
&
+
(1-a^2)\cdot
\iverson{\text{$(i+k)$ odd}}
+
a^2\cdot
\iverson{\text{$(i+k)$ even}}
.
\end{align}
To remove the dependence on
the parity of $i+k$,
the term
$
(1-a^2)\cdot
\iverson{\text{$(i+k)$ odd}}
+
a^2\cdot
\iverson{\text{$(i+k)$ even}}
$
must be constant.
It suffices to make
$1-a^2 = a^2$.
Thus,
$a = \frac{1}{\sqrt{2}}$.
Consequently,
$
\frac{1}{2}\cdot
\iverson{\text{$(i+k)$ odd}}
+
\frac{1}{2}\cdot
\iverson{\text{$(i+k)$ even}}
=
\frac{1}{2}
$
which
yields:
\begin{align}
\langle
\check{\mathbf{r}}_i
,
\check{\mathbf{r}}_k
\rangle_{\frac{N}{2}+1}
=
&
\frac{N}{4}
\cdot
\iverson{i+k\equiv 0 \pmod{N}}
+
\frac{N}{4}
\cdot
\iverson{i\equiv k \pmod{N}}
+
\frac{1}{2}
\label{equation-inner-prod-ext-bias}
.
\end{align}

This time the trick gave some result.
Instead of the check pattern,
the
inner product
possesses
a constant bias of $\frac{1}{2}$.
For $N=16$,
we have the following
inner product table:
\begin{align}
\begin{array}{c|cccc cccc c}
\langle \check{\mathbf{r}}_i, \check{\mathbf{r}}_k\rangle_{\frac{N}{2}+1}
&
\check{\mathbf{r}}_0 & \check{\mathbf{r}}_1 & \check{\mathbf{r}}_2 & \check{\mathbf{r}}_3 & \check{\mathbf{r}}_4 & \check{\mathbf{r}}_5 & \check{\mathbf{r}}_6 & \check{\mathbf{r}}_7 & \check{\mathbf{r}}_8
\\
\hline
\check{\mathbf{r}}_0
&
\frac{N}{2}+\frac{1}{2} & \frac{1}{2} & \frac{1}{2} & \frac{1}{2} & \frac{1}{2} & \frac{1}{2} & \frac{1}{2} & \frac{1}{2} & \frac{1}{2}
\\
\check{\mathbf{r}}_1
&
\frac{1}{2} & \frac{N}{4}+\frac{1}{2} & \frac{1}{2} & \frac{1}{2} & \frac{1}{2} & \frac{1}{2} & \frac{1}{2} & \frac{1}{2} & \frac{1}{2}
\\
\check{\mathbf{r}}_2
&
\frac{1}{2} & \frac{1}{2} & \frac{N}{4}+\frac{1}{2} & \frac{1}{2} & \frac{1}{2} & \frac{1}{2} & \frac{1}{2} & \frac{1}{2} & \frac{1}{2}
\\
\check{\mathbf{r}}_3
&
\frac{1}{2} & \frac{1}{2} & \frac{1}{2} & \frac{N}{4}+\frac{1}{2} & \frac{1}{2} & \frac{1}{2} & \frac{1}{2} & \frac{1}{2} & \frac{1}{2}
\\
\check{\mathbf{r}}_4
&
\frac{1}{2} & \frac{1}{2} & \frac{1}{2} & \frac{1}{2} & \frac{N}{4}+\frac{1}{2} & \frac{1}{2} & \frac{1}{2} & \frac{1}{2} & \frac{1}{2}
\\
\check{\mathbf{r}}_5
&
\frac{1}{2} & \frac{1}{2} & \frac{1}{2} & \frac{1}{2} & \frac{1}{2} & \frac{N}{4}+\frac{1}{2} & \frac{1}{2} & \frac{1}{2} & \frac{1}{2}
\\
\check{\mathbf{r}}_6
&
\frac{1}{2} & \frac{1}{2} & \frac{1}{2} & \frac{1}{2} & \frac{1}{2} & \frac{1}{2} & \frac{N}{4}+\frac{1}{2} & \frac{1}{2} & \frac{1}{2}
\\
\check{\mathbf{r}}_7
&
\frac{1}{2} & \frac{1}{2} & \frac{1}{2} & \frac{1}{2} & \frac{1}{2} & \frac{1}{2} & \frac{1}{2} & \frac{N}{4}+\frac{1}{2} & \frac{1}{2}
\\
\check{\mathbf{r}}_8
&
\frac{1}{2} & \frac{1}{2} & \frac{1}{2} & \frac{1}{2} & \frac{1}{2} & \frac{1}{2} & \frac{1}{2}  & \frac{1}{2} & \frac{N}{2}+\frac{1}{2}
\\
\hline
\end{array}
\end{align}

\paragraph{Removal of Constant Bias.}
A constant bias can be removed by scaling the values of
the basis functions at $n=0$.
This can be done because
$r_i[0]=1$, for all $i$.
So let us again consider another modification to the basis functions,
as follows:
\begin{align}
\label{equation-trick-1-a}
\hat{r}_i[n]
=
\begin{cases}
b \cdot r_i[0], & \text{if $n=0$,}
\\
r_i[n], & \text{if $n = 1, 2, \ldots, \frac{N}{2}-1$,}
\\
\frac{1}{\sqrt{2}}
\cdot
r_i[n], & \text{if $n =\frac{N}{2}$,}
\end{cases}
\end{align}
where
$b$ is a constant that must compensate the $\frac{1}{2}$ offset.
Therefore,
using~\eqref{equation-inner-prod-ext-bias},
the inner product
calculation
yields:
\begin{align}
\sum_{n=0}^{\frac{N}{2}}
\hat{r}_{i}[n]
\cdot
\hat{r}_{k}[n]
=
&
\hat{r}_{i}[0]
\cdot
\hat{r}_{k}[0]
+
\sum_{n=1}^{\frac{N}{2}}
\hat{r}_{i}[n]
\cdot
\hat{r}_{k}[n]
\\
=&
b^2
\cdot
r_{i}[0]
\cdot
r_{k}[0]
+
\left[
-1
+
\sum_{n=0}^{\frac{N}{2}}
\check{r}_{i}[n]
\cdot
\check{r}_{k}[n]
\right]
\\
=&
b^2
-1
+
\frac{N}{4}
\cdot
\iverson{i+k\equiv 0 \pmod{N}}
+
\frac{N}{4}
\cdot
\iverson{i\equiv k \pmod{N}}
+
\frac{1}{2}
\\
=&
b^2
-
\frac{1}{2}
+
\frac{N}{4}
\cdot
\iverson{i+k\equiv 0 \pmod{N}}
+
\frac{N}{4}
\cdot
\iverson{i\equiv k \pmod{N}}
.
\end{align}
The above expression
tells us
that
the
orthogonality requires
that
$b^2-\frac{1}{2} = 0$;
therefore
$b = \frac{1}{\sqrt{2}}$.
Consequently,
the orthogonal basis is obtained
by the following
calculation:
\begin{align}
\hat{\mathbf{r}}_k
=
\operatorname{diag}
(
\alpha_0, \alpha_1, \ldots, \alpha_{\frac{N}{2}}
)
\cdot
\mathbf{r}_k
,
\quad
k=0,1,\ldots,\frac{N}{2}
,
\end{align}
where
\begin{align}
\alpha_n
=
\frac{1}{\sqrt{2}}
\cdot
\iverson{n\equiv 0 \pmod{\frac{N}{2}}}
+
\iverson{n\not\equiv 0 \pmod{\frac{N}{2}}}
.
\end{align}

Still considering $N=16$,
now
we have the following
inner product table:
\begin{align}
\begin{array}{c|cccc cccc c}
\langle \hat{\mathbf{r}}_i, \hat{\mathbf{r}}_k\rangle_{\frac{N}{2}+1}
&
\hat{\mathbf{r}}_0 & \hat{\mathbf{r}}_1 & \hat{\mathbf{r}}_2 & \hat{\mathbf{r}}_3 & \hat{\mathbf{r}}_4 & \hat{\mathbf{r}}_5 & \hat{\mathbf{r}}_6 & \hat{\mathbf{r}}_7 & \hat{\mathbf{r}}_8
\\
\hline
\hat{\mathbf{r}}_0
&
\frac{N}{2} & 0 & 0 & 0 & 0 & 0 & 0 & 0 & 0
\\
\hat{\mathbf{r}}_1
&
0 & \frac{N}{4} & 0 & 0 & 0 & 0 & 0 & 0 & 0
\\
\hat{\mathbf{r}}_2
&
0 & 0 & \frac{N}{4} & 0 & 0 & 0 & 0 & 0 & 0
\\
\hat{\mathbf{r}}_3
&
0 & 0 & 0 & \frac{N}{4} & 0 & 0 & 0 & 0 & 0
\\
\hat{\mathbf{r}}_4
&
0 & 0 & 0 & 0 & \frac{N}{4} & 0 & 0 & 0 & 0
\\
\hat{\mathbf{r}}_5
&
0 & 0 & 0 & 0 & 0 & \frac{N}{4} & 0 & 0 & 0
\\
\hat{\mathbf{r}}_6
&
0 & 0 & 0 & 0 & 0 & 0 & \frac{N}{4} & 0 & 0
\\
\hat{\mathbf{r}}_7
&
0 & 0 & 0 & 0 & 0 & 0 & 0 & \frac{N}{4} & 0
\\
\hat{\mathbf{r}}_8
&
0 & 0 & 0 & 0 & 0 & 0 & 0 & 0 & \frac{N}{2}
\\
\hline
\end{array}
\end{align}
The above inner product table
shows that
basis functions
$\hat{\mathbf{r}}_0$
and
$\hat{\mathbf{r}}_{\frac{N}{2}}$
have
twice as much energy~\cite[p.~60]{oppenheim1999discrete}
than any of the other basis functions.
The basis functions can be made of equal energy
($\frac{N}{4}$)
by means of a scaling by $\sqrt{\frac{N/4}{N/2}}=\frac{1}{\sqrt{2}}$.
Thus we have the following basis functions:
\begin{align}
\label{equation-basis-1}
\beta_k \cdot \hat{\mathbf{r}}_k
,
\end{align}
where
\begin{align}
\beta_k
=
\frac{1}{\sqrt{2}}
\cdot
\iverson{k\equiv 0 \pmod{\frac{N}{2}}}
+
\iverson{k\not\equiv 0 \pmod{\frac{N}{2}}}
.
\end{align}

\begin{align}
\begin{array}{c|cccc cccc c}
\langle \beta_i \cdot \hat{\mathbf{r}}_i, \beta_k \cdot \hat{\mathbf{r}}_k \rangle_{\frac{N}{2}+1}
&
\beta_0 \cdot \hat{\mathbf{r}}_0 &
\beta_1 \cdot \hat{\mathbf{r}}_1 &
\beta_2 \cdot \hat{\mathbf{r}}_2 &
\beta_3 \cdot \hat{\mathbf{r}}_3 &
\beta_4 \cdot \hat{\mathbf{r}}_4 &
\beta_5 \cdot \hat{\mathbf{r}}_5 &
\beta_6 \cdot \hat{\mathbf{r}}_6 &
\beta_7 \cdot \hat{\mathbf{r}}_7 &
\beta_8 \cdot \hat{\mathbf{r}}_8
\\
\hline
\beta_0 \cdot \hat{\mathbf{r}}_0
&
\frac{N}{4} & 0 & 0 & 0 & 0 & 0 & 0 & 0 & 0
\\
\beta_1 \cdot \hat{\mathbf{r}}_1
&
0 & \frac{N}{4} & 0 & 0 & 0 & 0 & 0 & 0 & 0
\\
\beta_2 \cdot \hat{\mathbf{r}}_2
&
0 & 0 & \frac{N}{4} & 0 & 0 & 0 & 0 & 0 & 0
\\
\beta_3 \cdot \hat{\mathbf{r}}_3
&
0 & 0 & 0 & \frac{N}{4} & 0 & 0 & 0 & 0 & 0
\\
\beta_4 \cdot \hat{\mathbf{r}}_4
&
0 & 0 & 0 & 0 & \frac{N}{4} & 0 & 0 & 0 & 0
\\
\beta_5 \cdot \hat{\mathbf{r}}_5
&
0 & 0 & 0 & 0 & 0 & \frac{N}{4} & 0 & 0 & 0
\\
\beta_6 \cdot \hat{\mathbf{r}}_6
&
0 & 0 & 0 & 0 & 0 & 0 & \frac{N}{4} & 0 & 0
\\
\beta_7 \cdot \hat{\mathbf{r}}_7
&
0 & 0 & 0 & 0 & 0 & 0 & 0 & \frac{N}{4} & 0
\\
\beta_8 \cdot \hat{\mathbf{r}}_8
&
0 & 0 & 0 & 0 & 0 & 0 & 0 & 0 & \frac{N}{4}
\\
\hline
\end{array}
\end{align}

Fortuitously
we have that
$\beta_i = \alpha_i$,
$i=0,1,\ldots,\frac{N}{2}$.
It is worth noticing
that their roles are fundamentally different.
On the one hand,
the scaling by $\alpha_n$
acts upon the time-samples of the basis functions
by changing the values at $n=0$ and $n=\frac{N}{2}$
so that orthogonality can be achieved.
On the other hand,
the scaling by $\beta_k$
has a comparatively less important role.
It is simply a scaling factor to ensure that
the basis functions have the same energy
and facilitate
orthonormalization.
For orthonormality,
it suffices to multiply
each
basis function
$\beta_k \cdot \hat{\mathbf{r}}_k$
by $\sqrt{\frac{4}{N}}$.

\subsection{Removing the Check Pattern II}
\label{sec-approach-2}

Although
the approach described in the previous section
could furnish an orthogonal basis,
the
length of the resulting basis functions
is
$\frac{N}{2}+1$.
Such size
does not divide
the DFT length of $N$.
So,
for example,
the 16-point DFT basis functions
effect
9-point real-valued basis functions.
If there is any advantage
in keeping the basis functions
at multiple sizes,
this is lost in the discussed approach.

By re-examining the
check pattern
shown in~\eqref{tab-checkboard},
we notice that the undesirable terms $1$
can be avoid
if
$i+k$ or $i-k$
are even.
There are only two ways to achieved so:
(i)~either $i$ and $k$ are both even
or
(ii)~$i$ and $k$ are both odd.

\subsubsection{Real, Half-sized, and Even?}

Considering the case where $i$ and $k$ are even,
let
$i = 2i'$
and
$k = 2k'$,
for
$i',k'=0,1,\ldots,\frac{N}{4}-1$.
Now
we re-examine~\eqref{eq-real-half-fourier-inner-product}
and
re-write
the inner product
as:
\begin{align}
\langle
\mathbf{r}_{2i'}
,
\mathbf{r}_{2k'}
\rangle_{\frac{N}{2}}
=
&
\sum_{n=0}^{\frac{N}{2}-1}
r_{2i'}[n]
\cdot
r_{2k'}[n]
\\
=
&
\sum_{n=0}^{\frac{N}{2}-1}
\cos\left(
\frac{2\pi}{N}(2i')n
\right)
\cdot
\cos\left(
\frac{2\pi}{N}(2k')n
\right)
\\
=
&
\sum_{n=0}^{\frac{N}{2}-1}
\cos\left(
\frac{2\pi}{N/2}i'n
\right)
\cdot
\cos\left(
\frac{2\pi}{N/2}k'n
\right)
,
\end{align}
where
$i',k'=0,1,\ldots,\frac{N}{4}-1$.
For $N=32$,
we have the inner product table shown below:
\begin{align}
\begin{array}{c|cccc cccc}
\langle
\mathbf{r}_{2i'}
,
\mathbf{r}_{2k'}
\rangle_{\frac{N}{2}}
&
\mathbf{r}_0 & \mathbf{r}_2 & \mathbf{r}_4 & \mathbf{r}_6 & \mathbf{r}_0 & \mathbf{r}_{10} & \mathbf{r}_{12} & \mathbf{r}_{14}
\\
\hline
\mathbf{r}_0
&
\frac{N}{2} & 0 & 0 & 0 & 0 & 0 & 0 & 0
\\
\mathbf{r}_2
&
0 & \frac{N}{4} & 0 & 0 & 0 & 0 & 0 & 0
\\
\mathbf{r}_4
&
0 & 0 & \frac{N}{4} & 0 & 0 & 0 & 0 & 0
\\
\mathbf{r}_6
&
0 & 0 & 0 & \frac{N}{4} & 0 & 0 & 0 & 0
\\
\mathbf{r}_8
&
0 & 0 & 0 & 0 & \frac{N}{4} & 0 & 0 & 0
\\
\mathbf{r}_{10}
&
0 & 0 & 0 & 0 & 0 & \frac{N}{4} & 0 & 0
\\
\mathbf{r}_{12}
&
0 & 0 & 0 & 0 & 0 & 0 & \frac{N}{4} & 0
\\
\mathbf{r}_{14}
&
0 & 0 & 0 & 0 & 0 & 0 & 0 & \frac{N}{4}
\\
\hline
\end{array}
\end{align}

The above inner product table looks promising.
However,
notice that we have $\frac{N}{4}$
candidate basis functions,
but each function
has length of $\frac{N}{2}$.
Let us try to ``shrink'' each candidate basis function
by keeping only the first half of each function.
In other words,
each basis function has now $\frac{N}{4}$ components only.
Consequently,
for
$i',k'=0,1,\ldots,\frac{N}{4}$,
we have:
\begin{align}
\langle
\mathbf{r}_{2i'}
,
\mathbf{r}_{2k'}
\rangle_{\frac{N}{4}}
=
&
\sum_{n=0}^{\frac{N}{4}-1}
r_{2i'}[n]
\cdot
r_{2k'}[n]
\\
=
&
\sum_{n=0}^{\frac{N}{4}-1}
\cos\left(
\frac{2\pi}{N}(2i')n
\right)
\cdot
\cos\left(
\frac{2\pi}{N}(2k')n
\right)
\\
=
&
\sum_{n=0}^{\frac{N}{4}-1}
\cos\left(
\frac{2\pi}{N/2}i'n
\right)
\cdot
\cos\left(
\frac{2\pi}{N/2}k'{2}n
\right)
.
\end{align}
The above expression
is identical to~\eqref{eq-real-half-fourier-inner-product},
except
for the fact that the signal length
is halved.
So we are back to
non-orthogonality
issues
discussed in Section~\ref{section-real-half}
(see ~\eqref{tab-checkboard}).
In particular,
for $N=32$,
the inner product pattern is shown below:
\begin{align}
\label{tab-checkboard-2}
\begin{array}{c|cccc cccc}
\langle
\mathbf{r}_{2i'}
,
\mathbf{r}_{2k'}
\rangle_{\frac{N}{4}}
&
\mathbf{r}_0 & \mathbf{r}_2 & \mathbf{r}_4 & \mathbf{r}_6 & \mathbf{r}_0 & \mathbf{r}_{10} & \mathbf{r}_{12} & \mathbf{r}_{14}
\\
\hline
\mathbf{r}_0
&
\frac{N}{4} & 1 & 0 & 1 & 0 & 1 & 0 & 1
\\
\mathbf{r}_2
&
1 & \frac{N}{8} & 1 & 0 & 1 & 0 & 1 & 0
\\
\mathbf{r}_4
&
0 & 1 & \frac{N}{8} & 1 & 0 & 1 & 0 & 1
\\
\mathbf{r}_6
&
1 & 0 & 1 & \frac{N}{8} & 1 & 0 & 1 & 0
\\
\mathbf{r}_8
&
0 & 1 & 0 & 1 & \frac{N}{8} & 1 & 0 & 1
\\
\mathbf{r}_{10}
&
1 & 0 & 1 & 0 & 1 & \frac{N}{8} & 1 & 0
\\
\mathbf{r}_{12}
&
0 & 1 & 0 & 1 & 0 & 1 & \frac{N}{8} & 1
\\
\mathbf{r}_{14}
&
1 & 0 & 1 & 0 & 1 & 0 & 1 & \frac{N}{8}
\\
\hline
\end{array}
\end{align}

If the basis functions are ``shrunk''
by removing the odd- or even-indexed components,
the result is
that
the ``X''-like pattern re-appears
(see~\eqref{tab-saltire}).
Again
we are back to the previous discussed case---no success.

\subsubsection{Real, Half-sized, and Odd!}

Now let us try the alternative route:
keeping the odd-indexed basis functions only.
The
inner product calculation
furnishes the following:
\begin{align}
\langle \mathbf{r}_{2i'+1}, \mathbf{r}_{2k'+1}\rangle_{\frac{N}{2}}
=
&
\sum_{n=0}^{\frac{N}{2}-1}
r_{2i'+1}[n]
\cdot
r_{2k'+1}[n]
\\
=
&
\sum_{n=0}^{\frac{N}{2}-1}
\cos\left(
\frac{2\pi}{N}(2i'+1)n
\right)
\cdot
\cos\left(
\frac{2\pi}{N}(2k'+1)n
\right)
\\
=
&
\sum_{n=0}^{\frac{N}{2}-1}
\cos\left(
\frac{2\pi}{N/2}\frac{2i'+1}{2}n
\right)
\cdot
\cos\left(
\frac{2\pi}{N/2}\frac{2k'+1}{2}n
\right)
\\
=
&
\frac{N}{4}
\cdot
\iverson{\frac{2i'+1}{2} - \frac{2k'+1}{2} \equiv 0 \pmod{\frac{N}{2}}}
+
\frac{N}{4}
\cdot
\iverson{\frac{2i'+1}{2} + \frac{2k'+1}{2} \equiv 0 \pmod{\frac{N}{2}}}\\
=
&
\frac{N}{4}
\cdot
\iverson{i' - k' \equiv 0 \pmod{\frac{N}{2}}}
+
\frac{N}{4}
\cdot
\iverson{i' + k' + 1 \equiv 0 \pmod{\frac{N}{2}}}
,
\end{align}
Because $i',k'=0,1,\ldots,\frac{N}{4}-1$,
the quantity
$i' + k' + 1$
ranges from $1$ to $\frac{N}{2}-1$.
Thus, the condition set in
the above second Iverson bracket
is never satisfied,
and the evaluation always returns zero.
Therefore, we have:
\begin{align}
\label{equation-sum-odd-n/2}
\sum_{n=0}^{\frac{N}{2}-1}
r_{2i'+1}[n]
\cdot
r_{2k'+1}[n]
=
&
\frac{N}{4}
\cdot
\iverson{i' - k' \equiv 0 \pmod{\frac{N}{2}}}
,
\quad
i',k'=0,1,\ldots,\frac{N}{4}-1
.
\end{align}
For $N=32$,
the inner product table is shown below:
\begin{align}
\begin{array}{c|cccc cccc}
\langle \mathbf{r}_{2i'+1}, \mathbf{r}_{2k'+1}\rangle_{\frac{N}{2}}
&
\mathbf{r}_1 & \mathbf{r}_3 & \mathbf{r}_5 & \mathbf{r}_7 & \mathbf{r}_9 & \mathbf{r}_{11} & \mathbf{r}_{13} & \mathbf{r}_{15}
\\
\hline
\mathbf{r}_1
&
\frac{N}{4} & 0 & 0 & 0 & 0 & 0 & 0 & 0
\\
\mathbf{r}_3
&
0 & \frac{N}{4} & 0 & 0 & 0 & 0 & 0 & 0
\\
\mathbf{r}_5
&
0 & 0 & \frac{N}{4} & 0 & 0 & 0 & 0 & 0
\\
\mathbf{r}_7
&
0 & 0 & 0 & \frac{N}{4} & 0 & 0 & 0 & 0
\\
\mathbf{r}_9
&
0 & 0 & 0 & 0 & \frac{N}{4} & 0 & 0 & 0
\\
\mathbf{r}_{11}
&
0 & 0 & 0 & 0 & 0 & \frac{N}{4} & 0 & 0
\\
\mathbf{r}_{13}
&
0 & 0 & 0 & 0 & 0 & 0 & \frac{N}{4} & 0
\\
\mathbf{r}_{15}
&
0 & 0 & 0 & 0 & 0 & 0 & 0 & \frac{N}{4}
\\
\hline
\end{array}
\end{align}

However,
the above is not yet the sought orthogonal basis.
The candidate basis contains $\frac{N}{4}$ functions,
but each function is $\frac{N}{2}$-point long.

\subsubsection{\ldots and Half-sized Again}

Now our goal is to ``shrink''  the size of the basis functions
by removing redundancy.
Notice that
$r_k[0] = 1$.
And for odd $k$,
we have
that
$r_k[N/4] = 0$ and $r_k[N/2-n] = -r_k[n]$.
Therefore, we have:
\begin{align}
\sum_{n=0}^{\frac{N}{2}-1}
r_{2i'+1}[n]
\cdot
r_{2k'+1}[n]
=
&
1
+
\sum_{n=1}^{\frac{N}{4}-1}
r_{2i'+1}[n]
\cdot
r_{2k'+1}[n]
+
0
+
\sum_{n=1}^{\frac{N}{4}-1}
r_{2i'+1}[N/2-n]
\cdot
r_{2k'+1}[N/2-n]
\\
=
&
1
+
\sum_{n=1}^{\frac{N}{4}-1}
r_{2i+1}[n]
\cdot
r_{2k+1}[n]
+
0
+
\sum_{n=1}^{\frac{N}{4}-1}
(-r_{2i'+1}[n])
\cdot
(-r_{2k'+1}[n])
\\
=
&
1+
2
\cdot
\sum_{n=1}^{\frac{N}{4}-1}
r_{2i'+1}[n]
\cdot
r_{2k'+1}[n]
\\
=
&
-1
+
2
\cdot
\sum_{n=0}^{\frac{N}{4}-1}
r_{2i'+1}[n]
\cdot
r_{2k'+1}[n]
,
\end{align}
where
$i',k'=0,1,\ldots,\frac{N}{4}-1$.
Therefore,
by applying~\eqref{equation-sum-odd-n/2},
we have that:
\begin{align}
\label{equation-sum-odd-n/4}
\sum_{n=0}^{\frac{N}{4}-1}
r_{2i'+1}[n]
\cdot
r_{2k'+1}[n]
=
\frac{N}{8}
\cdot
\iverson{i' - k' \equiv 0 \pmod{\frac{N}{2}}}
+
\frac{1}{2}
.
\end{align}

For $N=32$,
we have the following
inner product
table:
\begin{align}
\begin{array}{c|cccc cccc}
\langle \mathbf{r}_{2i'+1}, \mathbf{r}_{2k'+1}\rangle_{\frac{N}{4}}
&
\mathbf{r}_1 & \mathbf{r}_3 & \mathbf{r}_5 & \mathbf{r}_7 & \mathbf{r}_9 & \mathbf{r}_{11} & \mathbf{r}_{13} & \mathbf{r}_{15}
\\
\hline
\mathbf{r}_1
&
\frac{N}{8}+\frac{1}{2} & \frac{1}{2} & \frac{1}{2} & \frac{1}{2} & \frac{1}{2} & \frac{1}{2} & \frac{1}{2} & \frac{1}{2}
\\
\mathbf{r}_3
&
\frac{1}{2} & \frac{N}{8}+\frac{1}{2} & \frac{1}{2} & \frac{1}{2} & \frac{1}{2} & \frac{1}{2} & \frac{1}{2} & \frac{1}{2}
\\
\mathbf{r}_5
&
\frac{1}{2} & \frac{1}{2} & \frac{N}{8}+\frac{1}{2} & \frac{1}{2} & \frac{1}{2} & \frac{1}{2} & \frac{1}{2} & \frac{1}{2}
\\
\mathbf{r}_7
&
\frac{1}{2} & \frac{1}{2} & \frac{1}{2} & \frac{N}{8}+\frac{1}{2} & \frac{1}{2} & \frac{1}{2} & \frac{1}{2} & \frac{1}{2}
\\
\mathbf{r}_9
&
\frac{1}{2} & \frac{1}{2} & \frac{1}{2} & \frac{1}{2} & \frac{N}{8}+\frac{1}{2} & \frac{1}{2} & \frac{1}{2} & \frac{1}{2}
\\
\mathbf{r}_{11}
&
\frac{1}{2} & \frac{1}{2} & \frac{1}{2} & \frac{1}{2} & \frac{1}{2} & \frac{N}{8}+\frac{1}{2} & \frac{1}{2} & \frac{1}{2}
\\
\mathbf{r}_{13}
&
\frac{1}{2} & \frac{1}{2} & \frac{1}{2} & \frac{1}{2} & \frac{1}{2} & \frac{1}{2} & \frac{N}{8}+\frac{1}{2} & \frac{1}{2}
\\
\mathbf{r}_{15}
&
\frac{1}{2} & \frac{1}{2} & \frac{1}{2} & \frac{1}{2} & \frac{1}{2} & \frac{1}{2} & \frac{1}{2} & \frac{N}{8}+\frac{1}{2}
\\
\hline
\end{array}
\end{align}
The same phenomenon
shown in~Section~\ref{sec-real-half-ext-2}
takes place
again.
Due to the additive bias of $+\frac{1}{2}$,
the basis is not orthogonal.
To address this problem,
we resort to the same technique as before.
So
consider the following candidate basis:
\begin{align}
\label{equation-trick-1-b}
\tilde{r}_{k'}[n]
=
\begin{cases}
c \cdot r_{2k'+1}[0], & \text{if $n=0$,}
\\
r_{2k'+1}[n], & \text{if $n=1,2,\ldots,\frac{N}{4}-1$,}
\end{cases}
\end{align}
for $k'=0,1,\ldots,\frac{N}{4}-1$.
Now let us check the orthogonality
for this new tentative basis.
By applying~\eqref{equation-sum-odd-n/4},
we have:
\begin{align}
\langle
\tilde{\mathbf{r}}_{i'}
,
\tilde{\mathbf{r}}_{k'}
\rangle_{\frac{N}{2}}
=
\sum_{n=0}^{\frac{N}{4}-1}
\tilde{r}_{i'}[n]
\cdot
\tilde{r}_{k'}[n]
&
=
c^2
+
\sum_{n=1}^{\frac{N}{4}-1}
r_{2i'+1}[n]
\cdot
r_{2k'+1}[n]
\\
&=
c^2
-1
+
\left(
\frac{N}{8}
\cdot
\iverson{i' - k' \equiv 0 \pmod{\frac{N}{2}}}
+
\frac{1}{2}
\right)
,
\end{align}
for
$i',k'=0,1,\ldots,\frac{N}{4}-1$.
Thus,
to satisfy
the orthogonality property,
we must impose
$c^2-\frac{1}{2} = 0$;
therefore
$c =\frac{1}{\sqrt{2}}$.

Consequently,
the sought orthogonal basis is obtained
by the following
calculation:
\begin{align}
\label{equation-basis-2}
\tilde{\mathbf{r}}_{k'}
=
\operatorname{diag}
\left(
\gamma_0,
\gamma_1,
\ldots,
\gamma_{\frac{N}{4}-1}
\right)
\cdot
\mathbf{r}_{2k'+1}
,
\quad
k'=0,1,\ldots,\frac{N}{4}-1
,
\end{align}
where
\begin{align}
\gamma_k
=
\frac{1}{\sqrt{2}}
\cdot
\iverson{k\equiv 0 \pmod{\frac{N}{4}}}
+
\iverson{k\not\equiv 0 \pmod{\frac{N}{4}}}
.
\end{align}
The resulting inner product table is below:
\begin{align}
\begin{array}{c|cccc cccc}
\langle \tilde{\mathbf{r}}_{i'}, \tilde{\mathbf{r}}_{k'}\rangle_{\frac{N}{4}}
&
\tilde{\mathbf{r}}_0 &
\tilde{\mathbf{r}}_1 &
\tilde{\mathbf{r}}_2 &
\tilde{\mathbf{r}}_3 &
\tilde{\mathbf{r}}_4 &
\tilde{\mathbf{r}}_{5} &
\tilde{\mathbf{r}}_{6} &
\tilde{\mathbf{r}}_{7}
\\
\hline
\tilde{\mathbf{r}}_0
&
\frac{N}{8} & 0 & 0 & 0 & 0 & 0 & 0 & 0
\\
\tilde{\mathbf{r}}_1
&
0 & \frac{N}{8} & 0 & 0 & 0 & 0 & 0 & 0
\\
\tilde{\mathbf{r}}_2
&
0 & 0 & \frac{N}{8} & 0 & 0 & 0 & 0 & 0
\\
\tilde{\mathbf{r}}_3
&
0 & 0 & 0 & \frac{N}{8} & 0 & 0 & 0 & 0
\\
\tilde{\mathbf{r}}_4
&
0 & 0 & 0 & 0 & \frac{N}{8} & 0 & 0 & 0
\\
\tilde{\mathbf{r}}_{5}
&
0 & 0 & 0 & 0 & 0 & \frac{N}{8} & 0 & 0
\\
\tilde{\mathbf{r}}_{6}
&
0 & 0 & 0 & 0 & 0 & 0 & \frac{N}{8} & 0
\\
\tilde{\mathbf{r}}_{7}
&
0 & 0 & 0 & 0 & 0 & 0 & 0 & \frac{N}{8}
\\
\hline
\end{array}
\end{align}
If
the basis functions
are multiplied by
$\sqrt{\frac{8}{N}}$,
then
orthonormality is ensured.

\section{Discrete Transforms}

At the end of Sections~\ref{sec-approach-1}
and~\ref{sec-approach-2},
we have obtained two sets of orthogonal functions.
In this section,
we use them to derive
the associate linear transformations.
In the following,
the quantity $N$
still
refers
to the size of the DFT
which the obtained bases were derived from.

\subsection{Transformation I}

Let
$\mathbf{x} = \begin{bmatrix}x[0] & x[1] & \cdots & x[\frac{N}{2}]\end{bmatrix}\transp$
be a input vector.
The representation
of
$\mathbf{x}$
over the basis shown in~\eqref{equation-basis-1}
is given by the following coefficients:
\begin{align}
X_1[k]
=
\left\langle
\mathbf{x}
,
\sqrt{\frac{4}{N}}
\cdot
\beta_k
\cdot \hat{\mathbf{r}}_k
\right\rangle_{\frac{N}{2}+1}
&
=
\sqrt{\frac{4}{N}}
\cdot
\beta_k
\cdot
\sum_{n=0}^\frac{N}{2}
\alpha_n
\cdot
\cos\left(
\frac{2\pi}{N}kn
\right)
\cdot
x[n]
,
\quad
k = 0,1,\ldots,\frac{N}{2}
,
\end{align}
where
\begin{align}
\alpha_i
=
\beta_i
&
=
\frac{1}{\sqrt{2}}
\cdot
\iverson{i\equiv 0 \pmod{\frac{N}{2}}}
+
\iverson{i\not\equiv 0 \pmod{\frac{N}{2}}}
.
\end{align}
Now let us re-write the above expression
in terms of the
basis functions length:
$M = \frac{N}{2}+1$.
Thus we have:
\begin{align}
X_1[k]
&
=
\sqrt{\frac{2}{M-1}}
\cdot
\beta_k
\cdot
\sum_{n=0}^{M-1}
\alpha_n
\cdot
\cos\left(
\frac{\pi}{M-1}kn
\right)
\cdot
x[n]
,
\quad
k = 0,1,\ldots,M-1
,
\end{align}
where
\begin{align}
\alpha_i
=
\beta_i
&
=
\frac{1}{\sqrt{2}}
\cdot
\iverson{i\equiv 0 \pmod{M-1}}
+
\iverson{i\not\equiv 0 \pmod{M-1}}
.
\end{align}
In matrix terms,
we have
that the transform coefficients
$\mathbf{X}_1 = \begin{bmatrix}X_1[0] & X_1[1] & \cdots & X_1[M-1]\end{bmatrix}\transp$
are given by:
\begin{align}
\mathbf{X}_1
=
\mathbf{T}_1
\cdot
\mathbf{x}
,
\end{align}
where the transformation matrix is
\begin{align}
\mathbf{T}_1
&
=
\sqrt{\frac{4}{N}}
\cdot
\left[
\beta_k
\alpha_n
\cos\left(
\frac{2\pi}{N}kn
\right)
\right]_{k,n=0,1,\ldots,\frac{N}{2}}
\\
&
=
\sqrt{\frac{2}{M-1}}
\cdot
\left[
\beta_k
\alpha_n
\cos\left(
\frac{\pi}{M-1}kn
\right)
\right]_{k,n=0,1,\ldots,M-1}
.
\end{align}
By construction,
the above matrix satisfies orthogonality.
Therefore,
$\mathbf{T}_1\transp = \mathbf{T}_1^{-1}$.
Moreover,
if we interchange the roles of matrix element indexes
$k$ and $n$,
the matrix remains unaltered.
Consequently,
it is also true
that
$\mathbf{T}_1\transp = \mathbf{T}_1$;
thus
$\mathbf{T}_1$ is involutory~\cite[p.~363]{bernstein2018scalar},
i.e.
$\mathbf{T}_1^2 = \mathbf{I}$.
In short:
$\mathbf{T}_1 = \mathbf{T}_1\transp = \mathbf{T}_1^{-1}$.

Listing~\ref{lst-t1}
provides
\textsc{gnu} Octave~\cite{eaton2024gnu}
commands to generate $\mathbf{T}_1$.
The listings shown in this letter are
not to be understood
as ``production codes''~\cite[p.~5]{golub1996matrix},
but are simply meant for educational purposes.
In Listing~\ref{lst-t1},
the variables $\texttt{a}$ and $\texttt{b}$
(lines 3 and 4)
are
numerically identical;
however,
we have kept them
separated
to emphasize their different roles:
orthogonalization
and
normalization,
respectively.

\begin{algorithm}
\caption{Matrix $\mathbf{T}_1$}
\label{lst-t1}
\begin{lstlisting}
F = dftmtx(N);
R = real(F(1:end/2+1, 1:end/2+1));
a = [1/sqrt(2) ones(1, N/2-1) 1/sqrt(2)];
b = a;
T1 = sqrt(4/N) * diag(b) * R * diag(a);
\end{lstlisting}
\end{algorithm}

\subsection{Transformation II}

Let
$\mathbf{x} = \begin{bmatrix}x[0] & x[1] & \cdots & x[\frac{N}{4}-1]\end{bmatrix}\transp$
be a input vector.
Considering the basis shown in~\eqref{equation-basis-2},
the transform coefficients
are given by:
\begin{align}
X_2[k]
=
\left\langle
\mathbf{x}
,
\sqrt{\frac{8}{N}}
\cdot
\tilde{\mathbf{r}}_{k}
\right\rangle_{\frac{N}{4}}
=
&
\sqrt{\frac{8}{N}}
\cdot
\sum_{n=0}^{\frac{N}{4}-1}
\alpha_n
\cdot
\cos\left(
\frac{2\pi}{N}
\cdot
(2k+1)
\cdot
n
\right)
\cdot
x[n]
,
\quad
k = 0,1,\ldots,\frac{N}{4}-1
,
\end{align}
where
\begin{align}
\gamma_i
&
=
\frac{1}{\sqrt{2}}
\cdot
\iverson{i\equiv 0 \pmod{\frac{N}{4}}}
+
\iverson{i\not\equiv 0 \pmod{\frac{N}{4}}}
.
\end{align}
By considering the above expression
in terms of the transformation
own blocklength,
we have $M = \frac{N}{4}$
and
therefore:
\begin{align}
X_2[k]
=
&
\sqrt{\frac{2}{M}}
\cdot
\sum_{n=0}^{M-1}
\alpha_n
\cdot
\cos\left(
\frac{\pi}{2N}
\cdot
(2k+1)
\cdot
n
\right)
\cdot
x[n]
,
\quad
k = 0,1,\ldots,M-1
,
\end{align}
where
\begin{align}
\gamma_i
&
=
\frac{1}{\sqrt{2}}
\cdot
\iverson{i\equiv 0 \pmod{M}}
+
\iverson{i\not\equiv 0 \pmod{M}}
.
\end{align}
Re-casting the above expression
in matrix format
yields the transformation:
\begin{align}
\mathbf{X}_2
=
\mathbf{T}_2
\cdot
\mathbf{x}
,
\end{align}
where
$\mathbf{X}_2 = \begin{bmatrix}X_2[0] & X_2[1] & \cdots & X_2[M-1]\end{bmatrix}\transp$
is the vector of the transform-domain
coefficients
and
\begin{align}
\mathbf{T}_2
&
=
\sqrt{\frac{8}{N}}
\cdot
\left[
\alpha_n
\cos\left(
\frac{2\pi}{N}
\cdot
(2k+1)
\cdot
n
\right)
\right]_{k,n=0,1,\ldots,\frac{N}{4}-1}
\\
&
=
\sqrt{\frac{2}{M}}
\cdot
\left[
\alpha_n
\cos\left(
\frac{\pi}{2M}
\cdot
(2k+1)
\cdot
n
\right)
\right]_{k,n=0,1,\ldots,M-1}
.
\end{align}
Again
due to its very construction,
$\mathbf{T}_2\transp = \mathbf{T}_2^{-1}$.
However,
swapping the roles of $k$ and $n$
yields
a different matrix.
Therefore,
$\mathbf{T}_2\transp \neq \mathbf{T}_2$.
Matrix $\mathbf{T}_2\transp$
is a transformation on its own right.
Let us refer to it as $\mathbf{T}_3 = \mathbf{T}_2\transp$.
Its basic properties are:
\begin{align}
\mathbf{T}_3\transp &\neq \mathbf{T}_3
\\
\mathbf{T}_3\transp &= \mathbf{T}_2
\\
\mathbf{T}_3^{-1} &=\mathbf{T}_2
\\
\mathbf{T}_3 \cdot \mathbf{T}_3\transp &= \mathbf{I}
,
\end{align}
where
the identity matrix is $\mathbf{I}$.

Listing~\ref{lst-t2}
shows
\textsc{gnu}~Octave commands
to obtain matrices
$\mathbf{T}_2$ and $\mathbf{T}_3$.

\begin{algorithm}
\centering
\caption{Matrices $\mathbf{T}_2$ and $\mathbf{T}_3$}
\label{lst-t2}
\begin{lstlisting}
F = dftmtx(N);
R = real(F(2:2:end/2, 1:end/4));
g = [1/sqrt(2) ones(1, N/4-1)];
T2 = sqrt(8/N) * R * diag(g)
T3 = transpose(T2)
\end{lstlisting}
\end{algorithm}

\section{Discussion and Concluding Remarks}

\paragraph{Results.}
It might be noteworthy
the fact that the orthogonalization exercise
described in this letter
was not directed to find matrices
with any particular property,
except being real-valued and orthogonal.
Matrices $\mathbf{T}_1$ and $\mathbf{T}_2$
were not sought
to exhibit any particular
``good'' property or
to excel in any given performance metric.
The only goal was to find---somehow---an orthogonal set
of real basis functions out of the DFT basis functions.

\paragraph{Transforms.}
Surprisingly (or not?),
the resulting basis and matrices are quite interesting.
At this point,
the reader might have already
recognized that
the orthogonal transformations
$\mathbf{T}_1$,
$\mathbf{T}_2$,
and
$\mathbf{T}_3$
are
the
\mbox{DCT-I},
\mbox{DCT-III},
and
the very popular
\mbox{DCT-II},
respectively~\cite[p.~15]{rao1990discrete}.

\paragraph{Blocklengths.}
Throughout
the derivations in this letter,
it was tacitly
assumed
that $N$ is divisible by 2 or 4,
whenever necessary.

\paragraph{Tricks.}
In this letter,
we employed a number of tricks.
Such maneuvers
are not intended to be general, systematic procedures.
They are just a ``tour de force''
aiming at orthogonalization
``at any cost''.
In fact,
they are highly dependent
on the DFT matrix symmetrical structure.

However,
one of the tricks
consists of
the removal of a constant bias from the inner product pattern;
we used it in~\eqref{equation-trick-1-a}
and~\eqref{equation-trick-1-b}.
Because it worked twice%
\footnote{``If it works once, it's a trick.
If it works twice, it's a method. [\ldots]''~\cite[p.~442]{boyd2001chebyshev}.},
let us
formalize it in the proposition
below.

We use the following notation:
the vector of ones is $\mathbf{1}$;
the matrix of ones is $\mathbf{J}=\mathbf{1}\cdot\mathbf{1}\transp$;
matrix $\mathbf{D}$ is diagonal with non-zero diagonal elements;
and
subscripts
inform matrix dimension
for clarity.

\begin{proposition}
Let
$\mathbf{A}_{N \times N}$
be
a
real
matrix
such that
$
\mathbf{A}
=
\left[
\begin{array}{c|c}
b
\cdot
\mathbf{1}
&
\mathbf{B}
\end{array}
\right]
$,
where
$b\neq0$.
If
$
\mathbf{A}
\cdot
\mathbf{A}\transp
=
\mathbf{D}
+
c
\cdot
\mathbf{J}
$,
$c\neq0$,
then
$
\left[
\begin{array}{c|c}
\pm \sqrt{b^2-c}
\cdot
\mathbf{1}
&
\mathbf{B}_{N\times(N-1)}
\end{array}
\right]
$
is orthogonal.
\end{proposition}

\begin{proof}
The outer product rule~\cite[p.~287]{bernstein2018scalar}
tells us that
$
\mathbf{A}
\cdot
\mathbf{A}\transp
=
(b \cdot \mathbf{1})
\cdot
(b \cdot \mathbf{1})\transp
+
\mathbf{B}\cdot\mathbf{B}\transp
=
b^2 \cdot \mathbf{J}
+
\mathbf{B}\cdot\mathbf{B}\transp
$.
Let
$
\tilde{\mathbf{A}}_{N\times N}
=
\left[
\begin{array}{c|c}
a \cdot b \cdot \mathbf{1}
&
\mathbf{B}_{N\times(N-1)}
\end{array}
\right]
$,
where
$a\neq0$ is a real number.
Thus:
\begin{align}
\tilde{\mathbf{A}}
\cdot
\tilde{\mathbf{A}}\transp
&
=
(ab)^2
\cdot
\mathbf{J}
+
\mathbf{B}\cdot\mathbf{B}\transp
\\
&
=
(ab)^2
\cdot
\mathbf{J}
+
\left(
\mathbf{A}
\cdot
\mathbf{A}\transp
-
b^2
\cdot
\mathbf{J}
\right)
\\
&
=
(a^2-1)\cdot b^2
\cdot
\mathbf{J}
+
\mathbf{A}
\cdot
\mathbf{A}\transp
\\
&
=
(a^2-1)\cdot b^2
\cdot
\mathbf{J}
+
\mathbf{D}
+
c
\cdot
\mathbf{J}
\\
&
=
\left[
(a^2-1)\cdot b^2+c
\right]
\cdot
\mathbf{J}
+
\mathbf{D}
.
\end{align}
Therefore,
by setting
$a = \pm \sqrt{1 - \frac{c}{b^2}}$,
the matrix
$\tilde{\mathbf{A}}$
becomes orthogonal.
\end{proof}

\begin{corollary}
Let
$\mathbf{A}_{N\times(N-1)}$
be
a
real
matrix.
If
$
\mathbf{A}
\cdot
\mathbf{A}\transp
=
\mathbf{D}
+
c
\cdot
\mathbf{J}
$,
$c\neq0$,
then
$
\left[
\begin{array}{c|c}
\pm\sqrt{-c}
\cdot
\mathbf{1}
&
\mathbf{A}_{N\times(N-1)}
\end{array}
\right]
$
is orthogonal.
\end{corollary}
\begin{proof}
Consider the augmented matrix
$
\tilde{\mathbf{A}}_{N\times N}
=
\left[
\begin{array}{c|c}
a \cdot \mathbf{1}
&
\mathbf{A}_{N\times(N-1)}
\end{array}
\right]
$.
Therefore,
we have:
\begin{align}
\tilde{\mathbf{A}}
\cdot
\tilde{\mathbf{A}}\transp
=
a^2 \cdot \mathbf{J}
+
\mathbf{A} \cdot \mathbf{A}\transp
=
a^2 \cdot \mathbf{J}
+
\mathbf{D} + c \cdot \mathbf{J}
=
(a^2+c) \cdot \mathbf{J}
+
\mathbf{D}
.
\end{align}
If $a^2+c=0$, then $\tilde{\mathbf{A}}$ is orthogonal.
As long as $c<0$, the method above keeps the resulting matrix real-valued.
\end{proof}

{\small
\singlespacing
\bibliographystyle{siam}
\bibliography{ref}
}

\end{document}